%% file: Half-Integral-updated.tex
\documentclass[11pt]{article}
\usepackage[utf8]{inputenc}
\usepackage[T1]{fontenc}
\usepackage[sc]{mathpazo}
\usepackage{microtype}
\usepackage{amsmath}
\usepackage{amssymb}
\usepackage{amsthm}
\usepackage{bm,float}
\usepackage{dsfont}
\usepackage{authblk}
\usepackage{fullpage}
\usepackage{comment}
\usepackage{tikz}
\usepackage{mathtools}
\usepackage[shortlabels]{enumitem}
\usepackage{complexity}
\usepackage{bbm}
\usepackage{hyperref}
\usepackage[capitalize, nameinlink]{cleveref}
\usepackage{nag}

\usepackage{algorithm,algpseudocode}
\usepackage{multicol}
\usepackage[scaled]{helvet} 
\usepackage{thmtools}

\newcommand{\declarecolor}[2]{\definecolor{#1}{RGB}{#2}\expandafter\newcommand\csname #1\endcsname[1]{\textcolor{#1}{##1}}}
\declarecolor{White}{255, 255, 255}
\declarecolor{Black}{0, 0, 0}
\declarecolor{LightGray}{216, 216, 216}
\declarecolor{Gray}{127, 127, 127}
\declarecolor{Orange}{237, 125, 49}
\declarecolor{LightOrange}{251,229, 214}
\declarecolor{Yellow}{255, 192, 0}
\declarecolor{LightYellow}{255, 242, 200}
\declarecolor{Blue}{91, 155, 213}
\declarecolor{LightBlue}{222, 235, 247}
\declarecolor{Green}{112, 173, 71}
\declarecolor{LightGreen}{226, 240, 217}
\declarecolor{Navy}{68, 114, 196}
\declarecolor{LightNavy}{218, 227, 243}

\hypersetup{
	colorlinks=true,
	pdfpagemode=UseNone,
	citecolor=Green,
	linkcolor=Navy,
	urlcolor=Black,
	pdfstartview=FitW
}

\newcommand{\declareperson}[1]{\expandafter\newcommand\csname#1\endcsname[1]{\textcolor{Orange}{#1: ##1}}}

\declareperson{Anna}
\declareperson{Shayan}
\declareperson{Nathan}

\theoremstyle{plain}
\newtheorem{theorem}{Theorem}[section]
\newtheorem{lemma}[theorem]{Lemma}

\newtheorem{proposition}[theorem]{Proposition}
\newtheorem{fact}[theorem]{Fact}
\newtheorem{conjecture}[theorem]{Conjecture}
\newtheorem{assumption}[theorem]{Assumption}
\newtheorem{definition}[theorem]{Definition}

\theoremstyle{remark}
\newtheorem{remark}[theorem]{Remark}

\newlist{parts}{enumerate}{10}
\setlist[parts]{label=\arabic*.,ref=\arabic*}
\makeatletter
\if@cref@capitalise
	\crefname{partsi}{Part}{Parts}
\else
	\crefname{partsi}{part}{parts}
\fi
\makeatother
\Crefname{partsi}{Part}{Parts}

\usetikzlibrary{shapes, patterns, decorations, fit, intersections, arrows, automata, positioning}

\newcommand{\old}[1]{}

\newcommand{\st}{\text{s.t.}}

\renewcommand{\R}{\ensuremath{\mathbb R}}

\renewcommand{\P}[1]{{\mathbb{P}}\left[#1\right]}
\renewcommand{\PP}[2]{{\mathbb{P}}_{#1}\left[#2\right]}

\renewcommand{\E}[1]{{\mathbb{E}}\left[#1\right]}
\renewcommand{\EE}[2]{{\mathbb{E}}_{#1}\left[#2\right]}

\renewcommand{\path}[2]{{ S_{#1}, \ldots, S_{#2} }}

\def\b1{{\bf 1}}
\def\1{{\bf 1}}

\def\bz{{\bf z}}

\def\cB{{\cal B}}

\def\eps{{\epsilon}}

\def\cost{c}

\def\setminus{-}
\def\R{\mathbb{R}}

\title{An Improved Approximation Algorithm for TSP \\
in  the Half Integral Case}
\author{Anna R. Karlin}
\affil{\small University of Washington\\ \textsf{karlin@cs.washington.edu}}
\author{Nathan Klein} 
\affil{\small University of Washington\\ \textsf{nwklein@cs.washinton.edu}}
\author{Shayan Oveis Gharan} 
\affil{\small University of Washington\\ \textsf{shayan@cs.washington.edu}}

\begin{document}

\maketitle
\begin{abstract}
We design a $1.49993$-approximation algorithm for the metric traveling salesperson problem (TSP) for instances in which an optimal solution to the subtour linear programming relaxation is half-integral. 
These instances received significant attention over the last decade due to a conjecture  of Schalekamp, Williamson and van Zuylen stating that half-integral LP solutions have the largest integrality gap over all fractional solutions. So, if the conjecture of Schalekamp et al. holds true, our result shows that the integrality gap of the subtour polytope is bounded away from $3/2$.
\end{abstract}

\newpage

\input{intro}

\input{prelim}

\input{algorithm}


\section{Probabilistic lemmas}\label{sec:probFacts}
In this section, we present three probabilistic lemmas which show that in every min-cut there is at least one good edge, (and in some there are even more).
This immediately proves \cref{thm:mainprobabilistic}.

Note the last cycle that we choose in step \ref{alg:step:cycle} of the algorithm has all edges even at last so we don't need to address it in this section. Furthermore, by \cref{fact:e+notindS}, $e^+$ does not belong to any critical cut.

\begin{restatable}[Bottom edge lemma]{lemma}{bottom}
\label{lem:bottom-highest}
Suppose that $e=(u,v)$ is a bottom edge. Then $e$ is good (where $p \ge 3/16$). \end{restatable}

\begin{figure}[!htbp]
\centering
\begin{tikzpicture}
\foreach \a/\b in {-3/1,0/0,3/1}{\draw[bleudefrance,fill=blizzardblue,thick] (\a,\b) circle (12pt);}
\node at (0,-0.7) {$S_{u,e}$};
\node at (3,0.3) {$S_{v,e}$};
\node [draw=none] at (-3.6,-0.5) () {$S_e$};
\draw[gray,xshift=-0.05cm,yshift=-0.15cm] (-3,1) to (-0.15,0.05); 
\draw[gray,xshift=0.05cm,yshift=--0.15cm] (-3,1) to (-0.15,0.05);

\draw[gray,xshift=-0.05cm,yshift=0.15cm] (0.15,0.05) to (3,1); 
\path[-] (0.15,0.05) edge[below,midway,thick,orange,xshift=0.05cm,yshift=-0.15cm] node {$e=(u,v)$} (3,1);

\path[-] (-3,1) edge[below,midway,thick,purple,xshift=-0.2cm,yshift=-0.05cm] node {$a$} (-4,2.5);
\path[-] (-3,1) edge[right,midway,thick,purple,xshift=-0.05cm,yshift=0.15cm] node {$b$} (-3,2.75);

\path[-] (3,1) edge[below,midway,thick,purple,xshift=0.2cm,yshift=0.05cm] node {$c$} (4,2.5);
\path[-] (3,1) edge[right,midway,thick,purple,xshift=0.05cm,yshift=0.15cm] node {$d$} (3,2.75);

\draw [black,line width=1.2pt, dashed] (0,1) ellipse (4.5 and 2);
\end{tikzpicture}

\caption{Illustration of edges in Bottom Edge Lemma.}
\label{fig:bottom}
\end{figure}

\begin{proof}
If $e$ is a bottom edge then $S_e$ ($= S_{(u,v)}$) is a cycle cut. By construction, when a tree on $S_e$ is
selected in step 8 of the algorithm, exactly one edge is chosen between every pair of adjacent
nodes in $E(S_e)$. So it suffices to consider the edges in $\delta(S_e)$. 
These divide up into two pairs of cycle partners connecting $S_e$ to $V \smallsetminus S_e$, say  $\{a, b\}$ and $\{c,d\}$. (See \cref{fig:bottom}.)
Then by \cref{lem:427}, setting $S_1:= \{a, b\}$ and $S_2:= \{c,d\}$,
we have $\P{|S_1 \cap T|= 1\text{ and } |S_2 \cap T| = 1} \ge 3/16$.
\end{proof}


\begin{restatable}[Top edge lemma]{lemma}{top}
\label{lem:top-good}
In a critical cut $\delta(S)$ with one edge that goes higher, of the remaining three edges in the cut, at least two are good with $p \ge \frac{1}{16}$.
\end{restatable}
\begin{proof}

\begin{figure}[!htbp]
\centering
\begin{tikzpicture}
\foreach \a/\b/\c in {-3/e/3,-1/a/2,1/b/2,3/c/2}{
\draw[bleudefrance,fill=blizzardblue,thick] (\a,0) circle (12pt);
	\path[-] (\a,0) edge[right,midway,thick,purple,yshift=0.5em] node {$\b$} (\a,\c);
}
\path[-] (-3,0) edge[above,midway,thick,goodgreen] node {$f$} (-1,0);
\path[-] (-3,0) edge[below,midway,thick,goodgreen,bend right=25] node {$g$} (1,0);
\path[-] (-3,0) edge[below,midway,thick,goodgreen,bend right=38] node {$h$} (3,0);

\draw [black,line width=1.2pt, dashed] (0,-0.4) ellipse (4.2 and 1.3);
\node[black] at (-3.8,0.6) {$\delta(S)$};

\draw [black,line width=1.2pt, dashed] (0,0) ellipse (5.2 and 3);
\end{tikzpicture}

\caption{Illustration of top lemma. The figure shows the case where all three of $S_f$, $S_g$ and $S_h$ have an edge that goes higher ($a,b,c$ respectively). It may be that some number of these nodes do not have an edge going higher.}
\label{fig:top}
\end{figure}

First, suppose that $e$ is the edge that goes higher from $\delta(S)$, and $f$, $g$ and $h$ are the other edges
in $\delta(S)$. 

If the other endpoint of one of these three edges, say $S_h$, has no edge that goes higher then $h$ is good. (See left side of \cref{fig:top}.) To see this, observe that we can condition on $\delta (S_h)$ being 
even which by \cref{lem:cut_even} has probability at least 13/27. Given this event,
$|\{f,g,h\}\cap T|$ is either even or odd. In either case,
the event $e \in T$ is an independent event that occurs with probability 1/2. Therefore,
$$\P{ |\delta(S_h) \cap T|\text{ even}} \cdot\P {e\text{ makes }\delta(S) \text{ even} ~|~|\delta(S_h) \cap T |\text{ even}} \ge \frac{13}{27}\cdot\frac{1}{2}.$$

Therefore, at least two of $f$,$g$ and $h$ are good if  at least two of $S_f$, $S_g$ and $S_h$ do not have an edge that goes higher.

Consider next the case that, say, $S_f$ and $S_g$ both have an edge that goes higher, but $S_h$ doesn't. As before, $h$ is good. We claim that one of $f$ and $g$ is also good.
To see this, since $e\in T$ is independent of $f,g,h \in T$, we 
observe using \cref{lem:420} that $$\P{e \in T \text{ and } |\{f,g,h\}\cap T|
 = 1} \ge \frac{1}{2} \cdot \frac{1}{2}.$$
Next, apply \cref{lem:correlation} to $a$ and $b$ to conclude that for one of $a$ and $b$, say $a$, 
$$\P{a\in T | e \in T} \ge \frac{1}{4}.$$
Therefore, as before, regardless of the even/odd status of $X\cup f$ after conditioning $e$ in and $f,g,h$ to 1, the cut can be fixed by $a$, meaning we have $p \ge \frac{1}{16}$.

Finally, if all three of $S_f$, $S_g$ and $S_h$
have an edge that goes higher, then again, condition $e$ in and $\{f,g,h\}\cap T|
 = 1$. Then 
apply \cref{lem:correlation} to $a$, $b$ and $c$ to conclude that for two of them, say $a$ and $b$, their probability of being in $T$ given that $e \in T$ is at least 1/4. 
Therefore, each can fix their corresponding cut ($\delta(S_f)$ for $a$ and $\delta(S_g)$ for $b$) and both $f$ and $g$ are good.
\end{proof}

\begin{lemma}
For every critical cut $S$, there is an edge $e\in\delta(S)$ such that $\P{e\text{ even at last}} \geq 1/27$.
\end{lemma}
\begin{proof}
First, if there is a critical set $S'$ ($S'\neq S)$ such that $|\delta(S)\cap\delta(S')|=2$, then by \cref{fact:2gohigherbottom}, there exist at least two bottom edges in $\delta(S)$ and by \cref{lem:bottom-highest} they are good and we are done. Furthermore, if there is an edge in $\delta(S)$ that goes higher then we are done by \cref{lem:top-good}.

	Otherwise, assume $S$ has at most one edge to any other critical set and no edge goes higher. Recall that $T^-=T\setminus e^+$ is always a spanning tree.
	Let $\mu'$ be the conditional measure where $|\delta(S)\cap T^-|=1$. Observe that in $\mu'$, first we sample a tree\footnote{Note $T$ is not a spanning tree in $G\smallsetminus S$, that is why we need to look at $T^-$.} in $G\smallsetminus S$, and then we {\em independently} add an edge in $\delta(S)$. Therefore,  by \cref{fact:updown}, $\PP{\mu'}{e\in T}\leq 1/2$. 
	Thus, for at least one edge  $e\in \delta(S)$, we have $1/4\leq \PP{\mu'}{e\in T}\leq 1/2$.  In the special case that $\P{|\delta(S)\cap T|=1}=0$, we just let $e$ be an arbitrary edge in $\delta(S)$. Say, $\{e\}= \delta(S)\cap \delta(S')$.
		It remains to prove that $\PP{\mu}{\delta(S),\delta(S') \text{ even}}>1/27$ since $S,S'$ are the last cuts of $e$. 
	\begin{eqnarray*}
		\P{\delta(S),\delta(S') \text{ even}} &=& 1- \P{\text{$\delta(S)$ or $\delta(S')$ odd}}\\
		&=& 1- \P{\delta(S) \text{ odd}} - \P{\delta(S') \text{ odd}} + \P{\delta(S),\delta(S') \text{ odd}}\\
		&\geq & 13/27 - \P{\delta(S) \text{ odd}} + \P{|T\cap \delta(S)|=1 \wedge \delta(S')\text{ odd}},
	\end{eqnarray*}
	by \cref{lem:cut_even}. First, note if $\P{|T\cap \delta(S)|=1}=0$, then  we get that $\P{|T\cap \delta(S)|=2}=1$ (since $\P{|\delta(S)\cap T|\geq 1}=1$, and $\E{|\delta(S)\cap T|}=2$). So, the RHS is $13/27$ and we are done.
	
	Otherwise, 
	\begin{eqnarray*} \P{|T\cap \delta(S)|=1 \wedge \delta(S')\text{ odd}} &=& \P{\delta(S')\text{ odd} \Big| |T\cap\delta(S)|=1}\cdot \P{|T\cap \delta(S)|=1}\\
		&\geq&  \frac14\cdot\P{|T\cap \delta(S)|=1}.
	\end{eqnarray*}
	The inequality is because edge $e$ can make $\delta(S')$ odd by being in/out of the tree and that has probability at least $1/4$.
Therefore,
	\begin{eqnarray*}	
		\P{\delta(S),\delta(S') \text{ even}}&\geq & 13/27 - \P{\delta(S)\text{ odd}} + \frac14 \P{|T\cap \delta(S)|=1}\\
		&=& 13/27 - \frac34\P{|T\cap \delta(S)|=1} - \P{|T\cap \delta(S)|=3}.
	\end{eqnarray*}
	Finally, by \cref{fact:e1}, the RHS attains its minimum value when $|T\cap \delta(S)|$ is sum of 4 Bernoullis with success probabilities $1,1/3,1/3,1/3$.
\end{proof}

\section{Proof of Main Theorem}



Recall that every good edge is even on top with probability at least $p$. The following statement is the main technical result of this section.
\begin{lemma}\label{prop:goodojoin}
There is a (random) feasible $O$-join solution such that for every good edge $e$,
$$\E{y_e} \leq 1/4-p/240,$$	
and for every bad edge $y_e=1/4$ with probability 1.
\end{lemma}
Before, proving the above statement we use it to prove \cref{thm:main}.
\begin{proof}[Proof of \cref{thm:main}]
Consider the trivial $O$-join solution $y'$ where $y'_e=1/2$ if $e$ is good and $y'_e=1/6$ otherwise. Note that this is a valid $O$-join by \cref{thm:mainprobabilistic}.
Now, define $z=\alpha y + (1-\alpha)y'$ for some $\alpha$ that we choose later. It follows that for any good edge $e$,
$$\E{z_e} \le \alpha(\frac{1}{4} - \frac{p}{240}) + (1-\alpha)\frac{1}{2},$$
and for a bad edge $f$:
$$\E{z_f} = \alpha\frac{1}{4} + (1-\alpha)\frac{1}{6}$$
So, for $p=\frac{1}{27}$ and $\alpha = \frac{2160}{2161}$ we obtain $\E{z_e}\leq 0.249962$. 
Since any edge $e$ is chosen in $T$ with probability $1/2$ (up to a $2^{-n})$ error), we pay at most $1/2+\E{z_e}$ for any edge $e$ whereas $x$ pays $1/2$. Therefore, we get a $0.749962/0.5$ approximation algorithm.
\end{proof}
So, in the rest of this section we prove \cref{prop:goodojoin}.
\paragraph{O-join construction for good edges:}
For each good  edge $e$, define $B_e$ to be an independent Bernoulli random variable which is 1 with probability $p/p_e$, where $p$ is the {\em lower bound} on the probability that {\em any} good edge is even on top,  and $p_e$ is the {\em actual} probability that $e$ is even on top.
If $e,f$ are bottom edge companions, then we let $B_f=B_e$ (with probability 1). 
Note that this still makes selection of $e,f$ independent of $B_e$ and any other edge of the graph. 

We then construct an $O$-join solution for each 1-tree $T$ using the following three step process:
\begin{enumerate}
\item Initialize $y_e := 1/4$ for each edge $e \in E$.
\item Next, {if $e$ is even at last in $T$ and $B_e=1$}, reduce $y_e$ by $r_e(T)$ where:
 $$
r_e(T):= \begin{cases}
\beta &\text{ if $e$ is a bottom edge.}\\
\tau_2 & \text{if $e=\{u,v\}$ is a good  top edge and there are exactly 2 good top edges}\\ & \text{ in both $\delta(S_u)$ and $\delta(S_v)$ that do not go higher.}\\
\tau_3 & \text{if $e$ is a good top edge that does not meet the previous criteria.}
\end{cases}
$$
$\beta,\tau_2,\tau_3$ are parameters we will set later. For now, we just assume $\tau_3\leq \tau_2\leq \beta \leq 1/12$. When $r_e(T) > 0$, we say that $e$ is \textbf{reduced}.

\item On each cut $C$ that is odd, let $\Delta(C):= \sum_{e \in C} r_e (T)$ be the amount by which edges on that cut were reduced in step  2, and let $G_C$ be the set of good edges on $C$ such that $C$ is one of their last cuts.  
Now, for an edge $e\in G_C$, let $C'$ (and $C$) be the last cuts of $e$.
Then increase $y_e$  by $\max\left\{\frac{\Delta(C)}{|G_C|}, \frac{\Delta(C')}{|G_{C'}|}\right\}$. 
Notice that in this case, since $C$ is one of $e$'s top cuts, $e$ is not even on top in $T$ and therefore is not reduced in step 2. 
\end{enumerate}

 By construction, this is a valid $O$-join solution since on every min-cut we began with at least 4 edges crossing every cut with $y_e = 1/4$ and then guaranteed that every reduction on an odd cut in step 2 was compensated for by a matching increase on that cut in step 3. (The ultimate gain of course will come from the fact that many cuts will be even and hence, there will not need to be an increase.)  
    All non-min-cuts have at least 6 edges on them, each of value at least 1/6 (after reduction) and are therefore satisfied in \eqref{eq:tjoinlp}.

In the rest of the proof, it is enough to show that for any good edge $e$,
 $\E{y_e}\leq 1/4-p/240.$ 
We complete the proof using the following two lemmas.

\begin{lemma}\label{lem:topedgejoin}
If $\beta \geq 5\tau_2/4$, then for any good top edge  $e=\{u,v\}$, 
$$\E{y_e}\leq \frac14 - p\min\{\tau_2-\beta/2,\tau_3-5\beta/12\}.$$
\end{lemma}
\begin{proof}
Let $C:= \delta (S_u)$ and $C':= \delta (S_v)$. Recall that, since $e$ is a top edge, $C$ and $C'$ are critical cuts. 
Let $H_C$ (resp. $H_{C'}$) be the good edges in $C$ (resp. $C'$) that go higher.  
Note that $|H_C|$  and $|H_{C'}|$ is either 0 or 1 by \cref{fact:2cycle}. We consider 3 cases:

\begin{description}
\item [Case (i): $|G_C| = |G_{C'}| = 2$.]
In the worst case, there is an edge, say $f\in H_C$ and an edge $g\in H_{C'}$.
If $f$ is a bottom edge, then
$$\E{r_f(T)} = \beta\cdot p\quad\text{ and }\quad
\P{C\text{ is odd} | f \text{ reduced}} = 1/2$$
since $$\P{\text{parity of }|f \cap T| = \text{ parity of }|(\delta(S_u) \smallsetminus f) \cap T|} = 1/2,$$
by \cref{fact:efinout}.
Therefore, the expected reduction in step 2 on $C$ is at most
$$ \E{r_f(T)}\cdot\frac{1}{|G_C|}\cdot \P{C\text{ is odd} | f \text{ reduced}} = 
\beta p\cdot \frac{1}{2}\cdot\frac{1}{2} = \frac{\beta\cdot p}{4}.$$ 
On the other hand, if $f$ is a top edge,
then the expected reduction in step 2 on $C$ is at most
$$ \E{r_f(T)}\cdot\frac{1}{|G_C|}\cdot \P{C\text{ is odd} | f \text{ reduced}} \le
\tau_2 p\cdot \frac{1}{2}\cdot\frac{5}{8}  \leq \frac{\beta p}{4},$$
where we use \cref{lem:420} and the fact that edge $f$ is independent of the rest of edges in $C$ to infer that , $\P{C\text{ is odd} | f \text{ reduced}}\leq 5/8$.
The same reasoning applies to $g$, so we get 

$$\E{y_e} \le  \frac{1}{4} - \tau_2p + 2\cdot \frac{\beta p}{4}.$$

\begin{figure}[!htbp]
\centering
\begin{tikzpicture}
\foreach \a/\b/\c in {-1/f/2,1/g/2}{
\draw[bleudefrance,fill=blizzardblue,thick] (\a,0) circle (12pt);
	\path[-] (\a,0) edge[right,midway,thick,purple,yshift=0.5em] node {$\b$} (\a,\c);
}
\path[-] (-1,0) edge[above,midway,thick,blue,xshift=-0.3em,yshift=0.1em] node {} (-2,0.25);
\path[-] (-1,0) edge[below,midway,thick,orange,xshift=-0.3em,yshift=-0.1em] node {} (-2,-0.25);
\path[-] (1,0) edge[above,midway,thick,blue,xshift=0.3em,yshift=0.1em] node {} (2,0.25);
\path[-] (1,0) edge[below,midway,thick,orange,xshift=0.3em,yshift=-0.1em] node {} (2,-0.25);
\path[-] (-1,0) edge[below,midway,thick,orange] node {$e$} (1,0);

\draw [black,line width=1.2pt, dashed] (0,-0.2) ellipse (2.5 and 1.1);
\node[black] at (-2.7,0.6) {$\delta(S_{e})$};

\end{tikzpicture}\qquad
\begin{tikzpicture}
\foreach \a/\b/\c in {-1/f/2,1/g/2}{
\draw[bleudefrance,fill=blizzardblue,thick] (\a,0) circle (12pt);
	\path[-] (\a,0) edge[right,midway,thick,purple,yshift=0.5em] node {$\b$} (\a,\c);
}
\path[-] (-1,0) edge[above,midway,thick,orange,xshift=-0.3em,yshift=0.1em] node {} (-2,0.25);
\path[-] (-1,0) edge[below,midway,thick,orange,xshift=-0.3em,yshift=-0.1em] node {} (-2,-0.25);
\path[-] (1,0) edge[above,midway,thick,blue,xshift=0.3em,yshift=0.1em] node {} (2,0.25);
\path[-] (1,0) edge[below,midway,thick,orange,xshift=0.3em,yshift=-0.1em] node {} (2,-0.25);
\path[-] (-1,0) edge[below,midway,thick,orange] node {$e$} (1,0);

\draw [black,line width=1.2pt, dashed] (0,-0.2) ellipse (2.5 and 1.1);
\node[black] at (-2.7,0.6) {$\delta(S_{e})$};

\end{tikzpicture}

\caption{Cases (i) is on the left and case (ii) is on the right. Orange edges are good. On the left, $y_e$ is reduced by $\tau_2$ with probability $p$ and pays half the burden on both sides; on the right $y_e$ is reduced by $\tau_3$ with probability $p$ and pays half the burden on one side and one third of the burden on the other.}
\label{fig:casebc}
\end{figure}

\item [Case (ii):  $|G_C|\geq 3$ or $|G_{C'}|\geq 3$ (or both).]
Again, in the worst case, there is an edge, say $f\in H_C$ and an edge $g\in H_{C'}$.
By the same reasoning as above, $\E{y_e}$ is largest if $f$ and $g$ are bottom
edges. In this case, the same calculation as above gives,
$$ \E{r_f(T)}\cdot\frac{1}{|G_C|}\cdot \P{C\text{ is odd} | f \text{ reduced}} \le  
\beta p\cdot \frac{1}{|G_C|}\cdot\frac{1}{2},$$ 
and similarly for $g$,
so
$$\E{y_e} \le  \frac{1}{4} - \tau_3 p + \frac{\beta p}{2}\left(\frac{1}{2} + \frac{1}{3}\right).$$
\item [Case (iii): $|H_C|+|H_{C'}|\leq 1$.]
In the worst case, $|H_C|=1$ and $|H_{C'}|=0$ and $f\in H_C$ is a bottom edge. 
Note that in this case we may have $G_{C'}=\{e\}$. But the advantage is that we never increase $y_e$ to fix $C'$ since $H_{C'}=\emptyset$.
Then,
$$\E{y_e} = \frac{1}{4} - \tau_3 p + \beta p\cdot \frac1{|G_C|}\cdot \frac12 \leq 1/4 -\tau_3 p +\frac{\beta p}{4}.$$
\end{description}
Note that if case (iii) does not happen, then by \cref{lem:top-good} we have $|G_C|,|G_{C'}|\geq 2$. So, either  (i) or (ii) will happen.
\end{proof}

\begin{lemma}\label{lem:bottomedgeojoin}
If $3\tau_3\leq 2\tau_2$, then for any (good) bottom edge $e$, 
$$\E{y_e}\leq 1/4 - p\min\{\beta/4, 3\beta/4-\tau_2, \beta-4\tau_2/3\}.$$
\end{lemma}
\begin{proof}
Say $f$ is the companion of $e$.
Let $S=S_e$ and $S'$ be the parent of $S$ in the hierarchy of critical cuts. 
Say the last cuts of $e$ (and $f$) are $C=\{e,f,a,b\}$ and $C'=\{e,f,g,h\}$.
In other words $a,b$ are partners and $g,h$ are partners.
Note that $|G_C|=|G_{C'}|=2$ because all edges $\{a,b,g,h\}$ go to the higher ciritical cut $S_e$. 

\begin{description}
\item [{Case (i)}:   $a$ and $g$ go higher than $S$.]
We have $a,g\in \delta(S')$. So, by \cref{fact:2cycle},  $S'$ is also a cycle cut.
This means that $b$ and $h$ are companions and $a$ and $g$ are cycle partner pairs on $\delta(S')$. (See \cref{fig:casebc}).
Edge $e$ has to increase to fix the cuts $C,C'$ whenever $a,b,g$ or $h$ are decreased and the corresponding cut is odd.  The expected increase in $y_e$ due to reductions on $a,b,g,h$ divides into two types.

\begin{figure}[!htbp]
\centering
\begin{tikzpicture}
\foreach \a/\b in {-3/1,0/0,3/1}{\draw[bleudefrance,fill=blizzardblue,thick] (\a,\b) circle (12pt);}
\node at (0,-0.7) {$S_{u,e}$};
\node at (3.5,0.4) {$S_{v,e}$};
\node [draw=none] at (-2.5,-1) () {$S$};
\node [draw=none] at (-4.5,-1.6) (){$S'$};
\draw[gray,xshift=-0.05cm,yshift=-0.15cm] (-3,1) to (-0.15,0.05); 
\draw[gray,xshift=0.05cm,yshift=--0.15cm] (-3,1) to (-0.15,0.05);
\draw [dashed,red, line width=1.1pt] (0.5,-0.5) arc (-20:170:2.2);
\node [draw=none,red] at (-0.6,1.9) () {$C$};
\draw [dashed,red,line width=1.1pt] (3.7,1) arc (0:250:0.7);
\node [draw=none,red] at (3.9,1.2) () {$C'$};

\draw[gray,xshift=-0.05cm,yshift=0.15cm] (0.15,0.05) to node [above] {$f$} (3,1); 
\path[-] (0.15,0.05) edge[below,midway,thick,orange,xshift=0.05cm,yshift=-0.15cm] node {$e=(u,v)$} (3,1);

\path[-] (-3,1) edge[below,midway,thick,purple,xshift=-0.2cm,yshift=-0.05cm] node {$a$} (-5,4);
\path[-] (-3,1) edge[right,midway,thick,purple,xshift=-0.05cm,yshift=0.15cm] node {$b$} (-3,2.75);

\path[-] (3,1) edge[below,midway,thick,purple,xshift=0.2cm,yshift=0.05cm] node {$g$} (5,4);
\path[-] (3,1) edge[right,midway,thick,purple,xshift=0.05cm,yshift=0.15cm] node {$h$} (3,2.75);

\draw [black,line width=1.2pt, dashed] (0,1) ellipse (4.5 and 2);
\draw [black,line width=1.2pt, dashed] (0,1.2) ellipse (5.5 and 3.5);

\end{tikzpicture}
\caption{Case (i) of \cref{lem:bottomedgeojoin}.}
\label{fig:casebc}
\end{figure}

We start by $b,h$: By \cref{fact:efinout} and that $B_b=B_h$, we know that $b$ and $h$ are always reduced at the same time. Furthermore, conditioned on $b$ (and $h$) being reduced, we have
$$ \text{parity of } |T\cap C| = \text{parity of }|T\cap C'|. $$
This is simply because $b$ (and $h$) are reduced only when they are  even at last which implies $|T\cap \{a,g\}|=1$.
So, we can fix the reduction of $b,h$ simultaneously when we increase $e$ (or $f$). In other words, it is enough to only take into account the expected increase of $y_e$ due to $b$, i.e.,
$$ \E{r_b(T)}\cdot \frac1{|G_C|} \cdot \P{|C\cap T| \text{ is odd} | b \text{ reduced}} =  
\beta p\cdot \frac{1}{2}\cdot\frac{1}{2} = \frac{\beta p}{4}.$$ 

Now, we calculate the expected increase due to $a,g$: We compute the charge due to $a$ and the same will hold for $g$. Again, by \cref{fact:efinout}, $b$ and $h$ are chosen independently of $a,g$, i.e., $\P{|T\cap C|\text{ odd} | a \text{ reduced}} \leq 1/2$. Therefore, the expected increase due to $a$ is
$$ \E{r_a(T)} \cdot \frac1{|G_C|} \cdot\P{|C\cap T|\text{ odd} | a\text{ reduced}}\leq \beta p\cdot \frac12\cdot \frac12=\frac{\beta p}{4}$$
using $\tau_3,\tau_2\leq \beta$.
Therefore, altogether,
$$\E{y_e} \le  \frac{1}{4} - \beta p + 3\cdot \frac{\beta p}{4}.$$

\item [Case (ii): Only one edge, say $a$, goes higher than $S$.]

In this case, by the same reasoning as above, we have
$$ \E{r_a(T)}\cdot\frac{1}{|G_C|}\cdot \P{|C\cap T| \text{ is odd} | a \text{ reduced}} \le  
\beta p\cdot \frac{1}{2}\cdot\frac{1}{2},$$ 
since $b$ is independent of $a$ and it can be chosen to correct the parity.

For the remaining three edges $\{b,g,h\}$, by \cref{lem:top-good} either two or three of $b,g$ and $h$ are good.
If two are good, then each has an expected reduction of at most $\tau_2 p$ or
three are good and each has an expected reduction of at most $\tau_3 p$.

Therefore,
altogether,
$$\E{y_e} \le  \frac{1}{4} - \beta p +  \frac{\beta p}{4} + \frac{p}{2}\cdot\max\{2\tau_2, 3\tau_3\} \leq \frac14 -\frac{3\beta p}{4} + p\tau_2,$$
by the assumption of the lemma.

\item [Case (iii): $a,g$ are companions and $b,h$ are companions (on the next critical cut).]
This case follows the same analysis as case (i) but gains because in this case $a,g$ are also reduced simultaneously. 

\item [Case (iv): No edge goes higher than $S$ and all $\{a,b,g,h\}$ are top edges.]
 Some number of these edges are good; if more than two are good we pay $4\tau_3$ (at most) with probability $p$ and otherwise we pay $2\tau_2$ (at most) with probability $p$ . Then:
$$\E{y_e} \le \frac{1}{4} - \beta p + \frac{p}{|G_C|}\max\{2\tau_2,4\tau_3\} \leq \frac14 - \beta p + \frac{4\tau_2 p}{3}.$$
\end{description}
To finish the proof we just need to argue that we exhausted all cases. By \cref{fact:2cutinter}, among $\{a,b,g,h\}$ at most two go higher. By  \cref{fact:cyclepartners}, from each pair of cycle partners, i.e., $\{a,b\}$ or $\{g,h\}$, at most one goes higher. Therefore, if case (i) does not happen, we have at most one that goes  higher. If (i), (ii) do not happen, then no edge goes higher. So, by \cref{fact:2gohigherbottom} either all four edges in $\{a,b,g,h\}$ are bottom edges, i.e., case (iii), or none are bottom edges, i.e., case (iv).
\end{proof}
To finish the proof of \cref{prop:goodojoin}, let $\beta=1/12$, $\tau_2=7/120$ and $\tau_3=7/180$ chosen to satisfy $\tau_3\leq\tau_2\leq\beta$, $\beta\geq 5\tau_2/4$ and $3\tau_3\leq2\tau_2$. Plugging in these numbers into \cref{lem:topedgejoin} and \cref{lem:bottomedgeojoin} we obtain that $\E{y_e}\leq 1/4-p/240$ for any good edge $e$ as desired.








\bibliographystyle{alpha}
\bibliography{tsp}

\end{document}

%% file: intro.tex
\section{Introduction} \label{sec:intro}

In an instance of the traveling salesperson problem (TSP) we are given a set of $n$ cities along with their pairwise symmetric distances. The goal is to find a Hamiltonian cycle of minimum cost. In the metric TSP problem, which we study here, the distances satisfy the triangle inequality. Therefore the problem is equivalent to finding a closed Eulerian connected walk of minimum cost. It is NP-hard to approximate TSP with a factor better than $\frac{185}{184}$ \cite{Lam12}. A classical algorithm of Christofides~\cite{Chr76} from 1976 gives a $\frac32$-approximation algorithm for TSP and remains the best known approximation algorithm for the general version of the problem despite significant work~\cite{Wol80,SW90,BP91,Goe95,CV00,GLS05,BM10,BC11,SWV12}. 

Polynomial-time approximation schemes (PTAS) have been found for Euclidean \cite{Aro96,Mitchell99}, planar \cite{GKP95, AGKKW98, Kle05},  and low-genus metric \cite{DHM07} instances. 
The case of graph metrics has received significant attention. In 2011, the third author, Saberi, and Singh~\cite{OSS11} found a $\frac{3}{2} - \epsilon_0$ approximation for this case. M\"omke and Svensson \cite{MS11} then
obtained a combinatorial algorithm for graphic TSP with an approximation ratio of 1.461. This approximation ratio was later improved by Mucha \cite{Muc12} to $\frac{13}{9} \approx 1.444$, and by Seb\"o and Vygen \cite{SV12} to $1.4$.

In this paper we study metric TSP for instances in which an optimal solution to the subtour linear programming relaxation is half-integral, i.e., when the optimal solution $x$  satisfies $x_e\in \{0,1/2,1\}$ for all edges $e$.
These instances are conjectured to be ``the hardest'' instances of TSP by Schalekamp, Williamson and van Zuylen.
\begin{conjecture}[\cite{SWvZ13}]\label{conj:SWZ}
The integrality gap for the subtour LP is attained on half-integral vertices of the polytope.
\end{conjecture}
The above conjecture is motivated by the fact that the worst  known  integrality gap examples of TSP (and TSP-path) are half-integral. Furthermore, as shown in \cite{SWvZ13} the worst case ratio of 2-matchings to optimal solution of the subtour-LP is attained by half-integral instances.
Very little progress has been made on half integral instances even though they have been a subject of study for decades, \cite{CR98, BC11, BS17,  HNR17, HN19}.


Our main result is the following theorem:
\begin{theorem}\label{thm:main} There is a randomized polynomial time algorithm which when given any half-integral fractional solution $x$ of the subtour LP produces a tour with expected cost at most $1.49993$ times the cost of $x$. 
\end{theorem}
So, if \cref{conj:SWZ} holds affirmatively, the above theorem implies that the integrality gap of the subtour-LP is at most $1.49993$.
Our result also has a direct consequence to the minimum cost 2-edge connected subgraph problem. About 20 years ago, Carr and Ravi \cite{CR98} showed that the integrality gap of the half-integral LP solutions of the min cost 2-edge connected subgraph problem is at most $4/3$. But, to the best of our knowledge, no polynomial time algorithm with an approximate factor better than $3/2$ is known. Our theorem also implies a $1.49993$ approximation algorithm for half-integral LP solutions of the minimum cost 2-edge connected subgraph problem.


\subsection{Overview of algorithm}
There are two well known lines of attack to metric TSP: (i) Start from an optimal Eulerian subgraph and make it connected by adding new edges while preserving the parity of the degrees, or (ii) Start with an optimal connected subgraph, then correct the parities of vertex degrees by adding the minimum cost Eulerian augmentation. 

Here, we take the second approach.
It turns out that in  approach (ii) the minimum cost Eulerian augmentation of any connected graph  is simply the min cost matching on odd degree vertices which can be computed in polynomial time. So, the main question is how to choose a spanning tree of cost at most OPT such that the cost of the minimum Eulerian augmentation is bounded away from $\text{OPT}/2$. Here we follow the approach initiated in \cite{OSS11}.
We sample a random spanning tree that does not cost more than the optimum in expectation. More precisely, we sample   from the maximum entropy distribution of spanning trees with marginals equal to the given LP solution, $x$, and then add the minimum cost matching on odd degree vertices. 

It was conjectured in \cite{OSS11} that this algorithm beats Christofides for general metric TSP, but the authors could only justify a variant of this conjecture for graph metrics. 
 To bound the approximation factor for graph metrics, \cite{OSS11} showed that this random spanning tree
``locally'' looks like a Hamiltonian cycle, e.g., each vertex has degree 2 with constant probability, and, except in some special cases, each pair of vertices  have degree 2 simultaneously with a constant probability.
Roughly speaking, the analysis of \cite{OSS11} is  ``local'' in the sense that it shows that there is a set $F$ of edges with $x(F)\geq \Omega(n)$, such that each $e\in F$ is only contained in a constant number of ``local'' (near) min cuts and all these (near) min cuts have even number of edges in the random spanning tree with constant probability.

Such a method  provably fails for the problem on general metrics since most of the cost of the LP may be concentrated on edges which show up in many (near) min cuts. So, one needs a more ``global'' analysis technique.
In this paper we take the first step towards a  global {\em amortized} analysis.
The hard instances of TSP are those where the cost of the LP is dominated by the edges which show up in many (near) min cuts; in this case, there is no hope to show that all such cuts are even simultaneously in a random spanning tree with constant probability. Our high level framework is to build  a hierarchy over edges. Roughly speaking, a more ``global'' edge shows up higher in the hierarchy. When an edge $e$ is even in its highest cuts in the hierarchy, we gain from $e$ at the expense of using descendants of $e$ in the hierarchy to pay for lower cuts of $e$ which may be odd. We then show that the amount $e$ gains when it reduces exceeds by a constant factor the amount it may have to pay to fix cuts containing edges going higher in the hierarchy.

Putting this together, we show that a variant of the max entropy sampling algorithm beats Christofides if the underlying LP solution is half-integral.
 We expect that many of our techniques can be generalized to apply to LP solutions that are not half-integral, however the most difficult barrier to overcome seems to be that the structure of near minimum cuts (all cuts in the LP of value within 2 and $2+\epsilon$ for a fixed $\epsilon > 0$) is more complex than the structure of minimum cuts.

%% file: prelim.tex
\section{Our Algorithm} \label{sec:prelim}
Before we discuss our algorithm, we need a few definitions and tools.
Where $V$ is the set of vertices, let $c:V\times V\to\R_+$ denote the cost of going from $u$ to $v$ for any $u,v\in V$.
For a graph $G=(V,E)$ and a set $S\subseteq V$, we write 
\begin{eqnarray*}E(S):=\{\{u,v\}\in E: u,v\in S\},\\
\delta(S):=\{\{u,v\}\in E: u\in S,v\notin S\}.	
\end{eqnarray*}
For a vector $x:E\to\R$, and a set $F\subseteq E$, we write $x(F)=\sum_{e\in F} x_e$.
For a graph $G=(V,E)$ and $S\subseteq V$, we write $G\smallsetminus S$ to be the graph where $S$ and all edges incident to $S$ are removed, and we write $G/S$ to denote the graph in which $S$ is contracted. For two sets $S,T$, we write $S\smallsetminus T$ to denote the set difference.

\subsection{Held-Karp Relaxation}
\label{sec:HK}
The following linear program was first formulated by Dantzig, Fulkerson and Johnson \cite{DFJ54} and is
known as the subtour elimination polytope or the Held-Karp LP relaxation (see also \cite{HK70}).
\begin{equation}
\begin{aligned}
\min & \sum\limits_{u,v} \cost(u,v) x_{\{u,v\}} \\
  \st &  \sum_{u \in S, v \in \overline{S}} x_{\{u,v\}}  \ge 2 & \forall \, S\subsetneq V\\
  & \sum_{v \in V} x_{\{u,v\}}  = 2~\hspace{6ex} & \forall \, u\in V\\
 &   x_{\{u,v\}}\geq  0  &  \forall \, u,v \in V.
\end{aligned}
\label{eq:tsplp}
\end{equation}

\begin{assumption}
Throughout the paper, we assume that we are given a feasible half-integral solution of the Held-Karp LP, that is, for each $\{u,v\}$, $x_{\{u,v\}} \in \{0, 0.5, 1\}$.
\end{assumption}

\begin{remark}
We will often talk about the support graph $G=(V,E)$ of $x$, replacing any edge of value 1 with two parallel edges.  Therefore the number of edges crossing any minimum cut is 4 (corresponding to fractional value 2),
and the graph is Eulerian. \textbf{Henceforth, any reference to the graph $G$ refers to this support graph.}
\end{remark}

In our algorithm we will repeatedly consider subsets $S$ of vertices such that $\delta (S)$ is a min-cut of $x$.


\subsubsection{Spanning Tree Polytope}
For any graph $G=(V,E)$,
Edmonds \cite{Edm70} gave the following description for the convex hull of spanning tree of a graph $\tilde G$, known as the {\em spanning tree polytope}.
\begin{equation}
\begin{aligned}
& z(E) = n-1 & \\
& z(E(S)) \leq |S|-1 &  \forall S\subseteq V\\
& z_e \geq 0 & \hspace{6ex} \forall e\in E.
\end{aligned}
\label{eq:spanningtreelp}
\end{equation}
Edmonds \cite{Edm70} proved that the extreme point solutions of this polytope are the characteristic vectors of the spanning trees of $G$.

We formally define tight sets in \cref{subsec:algorithm} but for now assume $S$ is tight if $x(\delta(S))=2$. In the half-integral case this corresponds to $|\delta(S)|=4$.
\begin{fact} If in the fractional solution $x$ of  the Held-Karp LP,
  the set $S$ is a tight set, then the restriction of $x$ to edges in $E(S)$, that is, the fractional solution $\{x_e\}_{e\in E (S)}$,
is in the spanning tree polytope on $(S, E(S))$. 
\end{fact}
\begin{proof}
Since every min-cut has  fractional value 2, and every vertex has fractional degree 2 in any Held-Karp  solution, we have
$$\sum_{e \in E(S)} x_e = \frac{2|S|-2}{2} = |S|-1.$$
 For the same reason, the constraint \eqref{eq:spanningtreelp} of the spanning tree polytope holds for each $S' \subset S$.
 \end{proof}

It follows immediately that:

\begin{fact} \label{fact:expcostT}If $S$ is a tight set w.r.t. $x$,  then
a random spanning tree on $S$ selected from a distribution whose marginals are $x_e$ for each $e\in E(S)$ has expected cost $\sum_{e\in E(S)} x_e c_e$.
 \end{fact}
 
 \subsection{Maximum entropy distribution}
 We say a distribution $\mu$ over spanning trees is 
{\em $\lambda$-uniform or maximum entropy} if there are nonnegative
weights $\lambda: E \rightarrow  \R_+$ such that for any tree $T$,  
$$ \P{T} \propto \prod_{e \in T} \lambda_e.$$

Given a point $z$ in the spanning tree polytope, for a connected graph $G = (V, E)$, Asadpour et al. \cite{AGMOS17} show that there is an efficient algorithm
that
finds non-negative $\lambda_e$'s in a  such a way that for every edge
$e \in E$ and tree $T$ sampled from $\mu$, $\P{e \in T}$ is
(approximately) equal to $z_e$.

To sample from a distribution on spanning trees, we follow \cite{AGMOS17,OSS11} and 
 sample spanning trees using a distribution

\begin{theorem}[\cite{AGMOS17}]
\label{thm:maxentropycomp}
Let $\bz$ be a point in the spanning tree polytope of the graph $\tilde G=(\tilde V, \tilde E)$.
For any $0 < \eps$, values $\lambda_e$ for all $e\in \tilde E$ can be found such
that
the corresponding $\lambda$-uniform spanning tree distribution, $\mu$, satisfies
$$
\sum_{T\in {\cal T}: T \ni e} \PP{\mu}{T}  \leq (1+\varepsilon)z_e,\hspace{3ex}\forall e\in E,$$
i.e., the marginals are approximately preserved.  Furthermore, the running
time is polynomial in $n=|\tilde V|$, $- \log \min_{e\in E} z_e$ and $\log(1/\eps)$.
\end{theorem}

We can now briefly explain the main algorithm of \cite{OSS11}: Given a feasible solution $x$ of subtour-LP, define $z=(1-1/n)x$. Then, sample $T$ from a $\lambda$-uniform distribution with marginals $z$ and add a min-cost matching on the odd degree vertices of $T$.


\subsection{Min-cost Eulerian augmentation}
\label{sec:Ojoin}
Once we have sampled a tree (or, as we shall see later, a tree plus an edge), we will be finding the minimum cost Eulerian augmentation. For this purpose, we use  the following characterization of the $O$-join polytope due to Edmonds and Johnson \cite{EJ73}.
\begin{proposition}
\label{prop:tjoin}
For any graph $G=(V,E)$, cost function $c: E \to \R_+$, and a set $O\subseteq V$ with an even number of vertices,  the minimum weight of an $O$-join equals the optimum value of the following integral linear program.
\begin{equation}
\begin{aligned}
\min \hspace{4ex} & \cost(y) \\
\st \hspace{3ex} & y(\delta(S)) \geq 1 & \forall S \subseteq V, |S\cap  O| \text{ odd}\\
& y_e \geq 0 & \forall e\in E
\end{aligned}
\label{eq:tjoinlp}
\end{equation}
\end{proposition}

\subsection{Description of Algorithm}\label{subsec:algorithm}
Our algorithm is a slight modification of the one studied in \cite{OSS11}.
Given a feasible solution $x$ of the subtour-LP, without loss of generality, we assume that there exists an edge $e^+$ such that $x_{e^+}=1$. If such an edge does not exist, we split a node $v$ into two nodes $v_1,v_2$; connect 2 of the edges out of $v$ to $v_1$ and the other two to $v_2$. Then, we connect $v_1$ to $v_2$ with edge $e^+$ of cost $c(e^+)=0$ and $x_{e^+}=1$.

Given such a solution $x$ our algorithm is as follows: Define a fractional spanning tree $z$ where $z_{e^+}=0$ and $z_e=x_e$ for any $e\neq e^+$. Then, we sample $T$ from the $\lambda$-uniform spanning tree distribution $\mu$ with marginals $z$ for some $\epsilon=2^{-n}$ using \cref{thm:maxentropycomp}. Define $T=T\cup\{e^+\}$; this gives a {\bf 1-tree}.
A 1-tree is a union of a spanning tree and an edge. Finally, we add a minimum cost Eulerian augmentation on the odd degree vertices of $T$. Throughout we let $T^-=T\smallsetminus {e^+}$. $T^-$ is an actual spanning tree.

\begin{figure}[htb]
\begin{center}
\begin{tikzpicture}[inner sep=1.7pt,scale=.8,pre/.style={<-,shorten <=2pt,>=stealth,thick}, post/.style={->,shorten >=1pt,>=stealth,thick}]
\draw [rotate=20,line width=1.1](0,0) ellipse (2cm and 1cm);
\draw [rotate=-20,line width=1.2](-1.35,-0.48) ellipse (2cm and 1cm);
\draw  (-3.8, 0.5) node {$X$};
\draw   (2.3, 0.5) node {$Y$};

\tikzstyle{every node} = [draw, circle,fill=red];
\node at (1.1,0.75)   (){};
\node at (0.6,0.3)   (){};
\path (1.3,0.1) node  (){};
\path (-1.9,0.35) node  (){};
\path (-2.5,0.5) node  (){};
\path (-0.5,0.2) node  (){};
\path (-1.2,-.7) node  (){};
\path (0,-0.7) node  (){};
\path (-1,1.5) node  (){};
\end{tikzpicture}
\end{center}
\caption{An example of two crossing sets.}
\label{fig:crossingsets}
\end{figure}

There is an equivalent description of the above algorithm. Before discussing this, we need to define three more concepts.
\begin{definition}
Consider a  graph  $G= (V,E)$ with min-cuts of value $k$.
\begin{itemize}
\item Any set $S\subseteq V$ such that $|\delta (S) |= k$ (i.e., its boundary is a min-cut) is called a \textbf{tight} set.
\item A cut $(S, \overline S)$ is \textbf{proper} if
$|S| \ge 2$ and $|\overline S| \ge 2$. 
\item Two sets $S$ and $S'$ \textbf{cross} if all of $S \smallsetminus S'$,
 $S' \smallsetminus S$, $S \cap S'$ and $V \smallsetminus (S \cup S')$ are non-empty.
\end{itemize}
\end{definition}
See \cref{alg} for an equivalent description of our algorithm, which we will work with throughout the paper. As the equivalence is not fundamental to our proof, we omit the (simple) proof here.
\begin{algorithm}[h]
\caption{Algorithm for half-integral TSP}
\begin{algorithmic}[1]
\State Given a half-integral solution $x$ of the subtour LP, with an edge $e^+$ with $x_{e^+}=1$.
\State Let $G$ be the support graph of $x$. 
\State Set $T = \emptyset$ \Comment{$T$ will be a 1-tree} 
\While {there exists a proper tight set of $G$ that is not crossed (by a tight set)}
\State Let $S$ be a minimal such set such that $e^+\notin E(S)$ \Comment{Note such a set always exists, as $S,\overline{S}$ are both proper tight sets, so one does not have $e^+$. In \cref{fact:e+notindS} we show that  $e^+\notin\delta(S)$.}\label{alg:step:mintightset}
\State Compute the maximum entropy distribution $\mu$ of $E(S)$
\State Sample a tree from $\mu$ and add its edges to $T$ 
\State Set $G = G/S$ \Comment{Note we never contract $e^+$.}
\EndWhile
\State Randomly sample a cycle from $G$ (including $e^+$) and add it to $T$ \Comment{In \cref{fact:lastcycle} we show $G$ itself is a cycle} \label{alg:step:cycle}
\State Compute the minimum O-Join on the odd nodes of $T$. Shortcut it and output the resulting Hamiltonian cycle.
\end{algorithmic}
\label{alg}
\end{algorithm}
 
\begin{fact}\label{fact:e+notindS}
In step \ref{alg:step:mintightset} of the algorithm, we have  $e^+\notin\delta(S)$.	
\end{fact}
\begin{proof}
Say $e^+\in\delta(S)$, and let $e^+=\{u,v\}$. Then, since $x_{e^+}=1$, $\{u,v\}$ is a tight set. It also crosses $S$ (as $S$ is a proper set). That is a contradiction.
\end{proof}

A few remarks are in order.
By \cref{fact:expcostT}, $\E{c(T)}=c(x)$ (up to an error of $2^{-n}$). Therefore, to prove \cref{thm:main} all we need to do is to bound the expected cost of the $O$-join by $c(x)(1/2-\eps_0)$ for some $\eps_0>0$.
Also, crucial to our analysis are the independence properties we get from our algorithm, see the following and \cref{fact:efinout}.
\begin{fact}\label{fact:independencecritical}
Any tree chosen from a max-entropy distribution corresponding to a proper tight set $S$ which is not crossed is independent of all other edges of $T$ that we choose in different iterations of the while loop in our algorithm.
\end{fact}

%% file: algorithm.tex
\section{Tools}
In this section we state two main tools that we use in our analysis, namely the cactus representation of minimum cuts and strongly Rayleigh probability distributions. We will conclude this section by giving an overview of our proof.
\subsection{Min-cuts and the cactus}\label{subsec:mincuts}

To understand our algorithm and analysis it is useful to recall the {cactus representation~\cite{FF09}} of the min-cuts of a graph.  We briefly recall some basic definitions and the recursive construction of the cactus. 
We will rely on a number of basic facts about min-cuts. {For proofs, see~\cite{FF09}}. Suppose $G$ is a $k$-edge connected graph.
\begin{fact}
\label{fact:min-cuts1}
 If two tight sets $S$ and $S'$ cross, then each of $S \smallsetminus S'$,
 $S' \smallsetminus S$, $S \cap S'$ and $\overline{S \cup S'}$ are tight. Moreover, there are no edges from $S \smallsetminus S'$ to
 $S' \smallsetminus S$, and there are no edges from 
 $S \cap S'$ to $\overline{S \cup S'}$.
 
 Therefore, if two distinct tight sets $S$ and $S'$ cross each other, then $\delta (S) \cap \delta (S') = \emptyset$.
 \end{fact}
 The following fact is especially useful to us, since the support graph of $x$ is 4-edge-connected.
 \begin{fact}
 \label{fact:min-cuts2}
Suppose that every proper mincut is crossed by some other proper mincut. Then $k$ is even and $G$ is a cycle, with $k/2$ parallel edges between each adjacent pair of vertices.
\end{fact}

\begin{fact}\label{fact:lastcycle}
In step \ref{alg:step:cycle} of the algorithm the remaining graph $G$ is a cycle of length at least 3 such that there are exactly two parallel edges between each pair of consecutive vertices.
\end{fact}
\begin{proof}
Let $G$ be the graph which remains after the while loop in the algorithm terminates. By the algorithm, $e^+=\{u,v\}$ is not contracted yet. $G$ has at least 3 vertices, as otherwise in the last of the while we contracted a set $S$ where $e^+\in \delta(S)$ which contradicts \cref{fact:e+notindS}. If $G$ has 3 vertices then it must be a cycle. Otherwise, $\{u,v\}$ is a proper tight set in $G$, and it must be crossed. In this case by \cref{fact:min-cuts2} $G$ is a cycle of length at least $4$.
\end{proof}

\begin{definition}[Cactus Graph] A loopless and 2-edge connected graph $C = (U,F)$
is a {\em cactus} if each edge belongs to exactly one cycle.
\end{definition}

\begin{theorem}[Dinits, Karzanov, Lomonosov] Let $G = (V, E)$ be a loopless graph with min-cut size $k \ge 1$. There is a cactus $C = (U,F)$
and a mapping $\phi: V \rightarrow U$ such that the 2-element cuts of
$C$ are in one to one correspondence with the min-cuts of $G$.
Equivalently, $S$ is at tight set of $G$ if and only if $\phi(X)$ is a tight set
of $C$.
\end{theorem}

Our algorithm can be viewed as essentially constructing a cactus representation of the min-cuts. More precisely, the critical cuts of our algorithm (defined below) are in one to one correspondence with the cycles of the cactus. 



In the rest of this section we will more fully explore the interaction between our algorithm and the structure of the cactus representation of minimum cuts. From now on, assume $G$ is the 4-regular support graph of the half-integral LP solution $x$ and that we have executed our algorithm on $G$. 

\paragraph{Critical sets and cuts:}
A tight set $S$ selected in step \ref{alg:step:mintightset} of the algorithm is called a \textbf{critical set} and the corresponding cut $\delta (S):= E(S, \overline S)$ is called a \textbf{critical cut}.
Vertices of $G$ are degenerate critical sets.

There is a natural {\em hierarchy of critical sets} associated with the execution
of the algorithm. The leaves of the hierarchy are vertices of the original graph.
If $S$ and $S'$ are critical sets such that $S$ or a contracted version of $S$ is a vertex in $S'$, then $S$ is a child  of $S'$ (respectively $S'$ is the parent  of $S$). If $\tilde S$ is an ancestor of $S$ in the hierarchy of critical sets, then we say that $\tilde S$ is a \textbf{higher} critical set
than $S$ (resp. $S$ is a \textbf{lower} critical set than  $\tilde S $). 
For example, in \cref{fig:exectree}, critical set $F$ is the parent of and is higher in the hierarchy than critical sets $A$, $B$ and $C$.

The root of
the hierarchy is the graph $G$ once we get to  step \ref{alg:step:cycle} of the algorithm. 
\begin{definition}[Going higher]\label{def:gohigher}
An edge $e$ in $\delta(S)$ {\em goes higher} if the lowest critical set $S'$ such that $S\subsetneq S'$ satisfies $e\in \delta(S')$.
\end{definition}
Note that by \cref{fact:independencecritical} any edge going higher is independent of all edges which do not.

\definecolor{goodgreen}{rgb}{0.0, 0.5, 0.0}
\definecolor{amethyst}{rgb}{0.6, 0.4, 0.8}
\definecolor{bleudefrance}{rgb}{0.19, 0.55, 0.91}
\definecolor{blizzardblue}{rgb}{0.67, 0.9, 0.93}
\tikzset{
	pics/Graph/.style n args={1}{
	code = {
	
	\ifthenelse{#1>0}{
		\node[state] (u1) {};
        \node[state] (a1) [right=1cm of u1] {};
        \node[state] (a2) [below=0.5cm of a1] {};
        \node[state] (a3) [below=0.5cm of a2] {};
        \path[-] (u1) edge[midway, thick,blue,below] node {$t$} (a1);
		\path[-] (u1) edge[thick, bend right] node {} (a2);
		\path[-] (u1) edge[bend right] node {} (a3);
    }{
        \node[state] (u1) {};
        \node[state] (a1) {};
        \node[state] (a2) {};
        \node[state,label=below:{$A$}] (a3) {};
    }
	
	\ifthenelse{#1>0}{
	\foreach \a\b in {b/1,c/4,d/7} {
		\node[state] (\a0) [below right=0.6cm and {\b cm} of a1] {};
		\node[state] (\a1) [right=1cm of \a0] {};
        		\node[state] (\a2) [below=0.5cm of \a0] {};
        		\node[state] (\a3) [below=0.5cm of \a1] {};
        }
        }{
        \def\offset{2}
        \foreach \a\b\d in {b/1/B,c/4/C,d/7/D} {
        		\foreach \c in {0,1,2} {
        			\node[state] (\a\c) [below right=0.6cm and {\b cm + \offset cm} of a1] {};
        		}
		\node[state,label=below:$\d$] (\a3) [below right=0.6cm and {\b cm + \offset cm} of a1] {};
        }
        }
        
        \ifthenelse{#1>0}{
        \node[state] (e1) [right=11cm of u1] {};
        \node[state] (e2) [below=0.5cm of e1] {};
        \node[state] (e3) [below=0.5cm of e2] {};}{
        \foreach \a in {1,2} {
        		\node[state] (e\a) [right=12.2cm of u1] {};
        }
        \node[state,label=below:$E$] (e3) [right=12.2cm of u1] {};
        }
        
        \node[state] (l1) [right=12.2cm of u1] {};
        
	\foreach \a in {b,c,d} {
		\foreach \b in {1,2,3} {
			\path[-] (\a0) edge node {} (\a\b);
		}
		\path[-] (\a1) edge node {} (\a2);
		\path[-] (\a1) edge node {} (\a3);
		\path[-] (\a2) edge node {} (\a3);
	}
	
	\foreach \a in {a,e} {
		\path[-] (\a1) edge node {} (\a2);
		\path[-] (\a1) edge[bend right] node {} (\a3);
		\path[-] (\a2) edge[black] node {} (\a3);
	}
	
	\path[-] (u1) edge[midway,thick,bend left=20,goodgreen] node {$a$} (l1);
	\path[-] (b1) edge node {} (c0);
	\path[-] (c3) edge[midway,thick,bleudefrance] node {$h$} (d2);
	\path[-] (a2) edge[midway,thick,amethyst] node {$e$} (b0);
	\path[-] (d1) edge node {} (e2);
	
	\path[-] (a1) edge[midway,thick,bend left=10,goodgreen] node {$g$} (d0);
	
	\path[-] (l1) edge node {} (e1);
	\path[-] (l1) edge[bend left] node {} (e2);
	\path[-] (l1) edge[bend left] node {} (e3);
	
	\ifthenelse{#1>0}{
	\path[-] (b3) edge node {} (c2);
	\path[-] (a3) edge[midway,thick,amethyst] node {$f$} (b2);
	\path[-] (d3) edge node {} (e3);	
	\path[-] (c1) edge[midway,thick,bend left=10,bleudefrance] node {$b$} (e1);
	}{
	\path[-] (b3) edge[bend right=10] node {} (c2);
	\path[-] (a3) edge[below,midway,thick,amethyst,bend right=10] node {$f$} (b2);
	\path[-] (d3) edge[bend right=10] node {} (e3);
	\path[-] (c1) edge[below,midway,thick,bend left=10,bleudefrance] node {$b$} (e1);
	}
	}
	}
}

\begin{figure}[htb!]
\centering
\begin{tikzpicture}[
            auto,
            node distance = 2.5cm, 
            semithick 
        ]

        \tikzstyle{every state}=[
            draw = black,
            thick,
            fill = white,
            minimum size = 1.5mm
        ]
	
	\pic at (0,0) {Graph={1}};
	\draw [purple,line width=1.2pt, dashed] (0.7,-0.85) ellipse (1.2 and 1.5);
	\node[purple] at (-1,-0.75) {$A$};
	\foreach \a/\b in {B/3.3, C/6.3, D/9.3} {
	\draw [purple,line width=1.2pt, dashed] (\b+0.1,-1.3) ellipse (1.2 and 1);
	\node[purple] at (\b-1.3,-2) {$\a$};
	}
	\draw [purple,line width=1.2pt, dashed] (11.9,-0.8) ellipse (1.3 and 1.5);
	\node[purple] at (11.9+1.3,-2) {$E$};
	
	\pic at (0,-4) {Graph={0}};
	\draw [purple,line width=1.2pt, dashed] (3,-4.5) ellipse (3.8 and 1.2);
	\node[purple] at (1.7-2.8,-4.5) {$F$};
	
	\draw [purple,line width=1.2pt, dashed] (11,-4.5) ellipse (2.4 and 1);
	\node[purple] at (11+2.8,-4.5) {$G$};
	
	\node[state,label=below:$F$] at (2,-7) (u1) {};
	\node[state,label=below:$G$] at (10.5,-7) (l1) {};
	\path[-] (u1) edge[goodgreen,midway,thick,bend left=15] node {$a$} (l1);
	\path[-] (u1) edge[below,goodgreen,midway,thick,bend left=10] node {$g$} (l1);
	\path[-] (u1) edge[below,midway,bleudefrance,thick,bend right=15] node {$b$} (l1);
	\path[-] (u1) edge[midway,thick,bleudefrance,bend right=10] node {$h$} (l1);

    \end{tikzpicture}
    \caption{Example execution on a half integral graph.
    In the first figure, we visualize five tree operations in parallel, which we may do since all these tight sets have size 4 and are not crossed by other tight sets. Similarly in the second figure we do two operations in parallel. In the final step, a cycle is chosen by picking two edges at random. $A,B,C,D,E$ are all ``degree cuts" whereas $F$ and $G$ are both ``cycle cuts." $t$ is an example top edge (as are all edges picked in the first graph). $a,g$ and $b,h$ are cycle partners with respect to the cut $F$. $e,f$ are companions. 
     } \label{fig:execution}
\end{figure}

\begin{figure}[htb!]
\centering
    \begin{tikzpicture}[level distance=1.5cm,
  level 1/.style={sibling distance=5cm},
  level 2/.style={sibling distance=1.5cm}]
  \node {Pick a cycle}
    child {node {F}
      child {node {A}}
      child {node {B}}
      child {node {C}}
    }
    child {node {G}
    child {node {D}}
      child {node {E}}
    };
\end{tikzpicture}
\caption{An execution tree on this graph.}	
\label{fig:exectree}
\end{figure}


\paragraph{Structure of critical cuts:} Consider a critical set $S$ chosen
in step \ref{alg:step:mintightset} in the algorithm. We will abuse notation and, at any time during the execution of the algorithm, refer to $G$ with vertex set $V$ as the graph remaining at that time, after contraction of all trees that have been sampled before $S$ is considered. Consider the graph $G' := G/ V \smallsetminus S$ and let $w$ be the contracted vertex representing
$V \smallsetminus S$. 
There are two possibilities for the structure of $G'$:
\begin{itemize}
\item Case 1: There are no proper min-cuts inside $S$. 
In this case, we call $\delta (S)$ a \textbf{degree cut}. In \cref{fig:execution}, $A,B,C,D,E$ are all degree cuts.\footnote{ These cuts correspond to cycles of length two in the cactus.}
\item Case 2: There is a proper min-cut $(S_0, \overline S_0)$ such that
$S_0 \subsetneq S$. In this case,
it (and every other proper min-cut inside $S$) is crossed by some other
min-cut (or would be more minimal than $S$). 

It follows that 
in $G'$,  every proper mincut is crossed by some other proper mincut and therefore, by \cref{fact:min-cuts2}, the graph is a cycle with two edges between each pair of adjacent vertices in the cycle.  In this case, we call $\delta(S)$ a {\em cycle cut}. For example, $F$ and $G$ in \cref{fig:execution} are cycle cuts. 

We divide the 4 edges from $w$ into two pairs, such that each pair share an endpoint  inside $S$. We call each such pair {\em cycle partners} with respect to $\delta(S)$.
Every other pair of edges between two adjacent vertices in the cycle are called {\em companions}.

For example, in \cref{fig:execution}, $\delta(F)$ is a cycle cut
and  $a,g$ and $b,h$ are cycle partners with respect to $\delta(F)$. $e$ and $f$ are companions.

Cycle cuts correspond to cycles of length 3 or more in the cactus. 
\end{itemize}
Note that every edge has at most one companion but possibly many partners depending on the underlying cactus.

\begin{definition}[Highest critical cuts]
For a vertex $u$ and an edge $e=\{u,v\}$, let $S_{u,e}$ be the highest critical set $S$ such that $u \in S$ and $v \not \in S$, and let $S_{e}$ be the lowest critical set such that both $S_{u,e}$ and $S_{v,e}$ are (contracted) nodes in $S_{e}$. Then $\delta(S_{u,e})$ and $\delta(S_{v,e})$ are the highest critical cuts containing $e$. 
If the edge $e$ is clear from context, we may drop $e$ in the notation $S_{u,e}$.
\end{definition}

\begin{definition}[Bottom Edge and Top Edges]
For an edge $e$, if $S_{e}$ is a cycle cut, we say that $e$ is a {\em bottom edge} and otherwise it is a {\em top edge}. 
\end{definition}
For example, in \cref{fig:execution}, $e,f,a,g,b,h$ are bottom edges (among the labeled edges) and $t$ is a top edge.
The following fact is immediate:
\begin{fact} 
Companion bottom edges $e,f$ are in or out of $T$ independently of every other edge of $T$.
\label{fact:efinout}
\end{fact}

\paragraph{Min-cuts containing a particular edge:} 
The set of min-cuts an edge $e=(u,v)$ is on are the following:
\begin{enumerate}
\item[(a)] all critical degree cuts $\delta(S)$ such that $e \in \delta(S)$. (This
includes the cuts $(u, V\smallsetminus u)$ and $(v, V\smallsetminus v)$.)
\item[(b)] 
For any set $S$ such that $\delta(S)$ is a critical cycle cut, and
$e$ is either in $S$ or on $\delta(S)$, every  cut of the cycle that includes
the edge $e$ is a min-cut $e$ is on.
\end{enumerate}

\old{
\begin{figure}\centering
\begin{tikzpicture}[scale=0.8,inner sep=1.8]
\tikzstyle{every node}=[draw,fill=red,circle];
\begin{scope}[shift={(-4,0)}]
\foreach \i in {0,...,5}{
\node at (\i*60:1.5) (a_\i) {};
}
\foreach \i/\j in {0/1, 1/2, 2/3, 3/4, 4/5, 5/0}{
\path (a_\i) edge (a_\j);
}
\end{scope}
\begin{scope}[shift={(4,0)}]
\foreach \i in {0,...,5}{
\node at (\i*60:1.5) (a_\i) {};
}
\foreach \i/\j in {0/1, 1/2, 2/3, 3/4, 5/0}{
\path (a_\i) edge (a_\j);
}
\draw [color=blue,dashed,line width=1.3] (a_5)+(-.3,-.3) arc (180:55:1);
\draw [color=blue,dashed,line width=1.3] (a_5)+(-.5,-.3) arc (180:120:3);
\draw [color=blue,dashed,line width=1.3,shift=(a_5)] (a_5)+(-.7,-.3) -- (-.7,3);
\draw [color=blue,dashed,line width=1.3] (a_5)+(-.9,-.3) arc (0:60:3);
\draw [color=blue,dashed,line width=1.3] (a_5)+(-1.1,-.3) arc (0:120:1);
\end{scope}
\end{tikzpicture}
\caption[An Example of a Graph with no Good Edge]{Consider the Hamiltonian cycle shown at the left. In any spanning tree of this graph
all edges are contained in at least one odd minimum cut. The dashed blue arcs in the right shows the odd near minimum cuts for one spanning tree.}
\label{fig:cyclegoodedges}
\end{figure}
}

It is easy to see that each of the above is a min-cut. To see that there are no others, it suffices to observe by induction that whenever a set $S$ is contracted,
we have accounted for all min-cuts in which nodes inside $S$ are partitioned between the two sides of the cut.\old{  this follows immediately from the fact that when a set $S$ is selected, either there are no proper
min-cuts inside $S$, or  the remaining graph with  $V \smallsetminus S$ contracted is a cycle with two edges between each pair of edges (and case (b) above accounts for all such cuts).}

\paragraph{Other facts:} 
We end this part by recording the following basic facts about structure of min cuts, and we will use them throughout our proofs.
\begin{fact} \label{fact:cyclecut} Suppose that $S$ is a critical set.
If some (contracted) vertex $v \in S$ has two edges to $w := V \smallsetminus S$, then $S$ is a cycle cut. 
\end{fact}
\begin{proof} This is immediate if $|S|=2$, so suppose that $|S|> 2$.
Then
$v$ has two edges to $w$, which has two edges to $S \smallsetminus v$ which has two edges to $w$. Since $S \smallsetminus v$ is therefore a proper min-cut but was not selected in step \ref{alg:step:mintightset}, it must be crossed by some other set, which, by the earlier discussion of the structure of critical cuts , means that $\delta(S)$  is a cycle cut.
\end{proof}

\begin{fact} \label{fact:2cutinter} Suppose that $S$  and $S'$ are two distinct tight sets. 
 Then $|\delta (S) \cap \delta (S')| \le 2$.
 \end{fact}
 \begin{proof}
 By contradiction. Suppose that $S$ and $S'$ are both proper min-cuts and have
 $\delta (S) \cap \delta (S') \ge 3 $. Then by \cref{fact:min-cuts1}, they do not cross. 
 Therefore it must be that, say, $S \subset S'$. But in this case, if $\delta (S) \cap \delta (S') \ge 3$, since $\delta(S)$ and $\delta(S')$ are both min-cuts, there can only be one edge from $S$ to $S'\smallsetminus S$ and at most one edge from $ S'\smallsetminus S$ to $V \smallsetminus (S \cup S')$ which contradicts $\delta (S' \smallsetminus S) \ge  4$.
 \end{proof}

\begin{fact} 
\label{fact:2cycle} Suppose that $S$  and $S'$ are two critical sets such that $S\subset S'$.
Then if $\delta (S) \cap \delta (S') = 2$, then $S'$ is a cycle cut.
\end{fact}
\begin{proof}
Once $S$ is contracted, it has two edges to $V \smallsetminus S'$, and therefore by
 \cref{fact:cyclecut} is a cycle cut. 
\end{proof}

\begin{fact}
\label{fact:cyclepartners}
Suppose that $S \subset S'$ are two critical cycle cuts. Then any two edges  are cycle partners on at most one of these (cycle) cuts. 
\end{fact}
\begin{proof}
Suppose not. Then there is a pair of edges
$e$ and $f$ that are cycle partners on both. Suppose that
$g$ and $h$ are the other pair of cycle partners on $\delta(S)$ and that their endpoint inside $S$ is node $u$.
Then $(S'\smallsetminus S) \cup u$ is a min-cut that crosses $S$, which is a contradiction to the selection of $S$.  [Essentially this means that in fact there is a larger cycle here.]
\end{proof}

\begin{fact}\label{fact:2gohigherbottom}
Say $S$ is a critical set and exactly two edges of $\delta(S)$ are bottom edges that do not go higher. Then the other two edges of $\delta(S)$ must go higher.
\end{fact}
\begin{proof}
Say $\delta(S)=\{a,b,c,d\}$ and suppose $a,b$ are bottom edges that do not go higher. Say $S'$ is the parent of $S$ in the hierarchy of critical cuts. This implies that $\delta(S')$ is a cycle cut. So, $a,b$ are companions in this cycle. This implies that either $c,d$ are also companions or they are cycle partners in $\delta(S')$.
\end{proof}

\subsection{Strongly Rayleigh Distributions}
\label{sec:SRprops}
Let $\cB_E$ be the set of all probability measures on the Boolean algebra $2^E$. 
Let $\mu\in\cB_E$. The generating polynomial $g_\mu: \R[\{y_{e}\}_{e\in E}]$ of $\mu$ is defined as follows:
$$ g_\mu(y):=\sum_S \mu(S) \prod_{e\in S} y_e.$$
We say $\mu$ is a strongly Rayleigh distribution if $g_\mu\neq 0$ over all $\{y_e\}_{e\in E} \in \mathbb{C}^E$ where $\text{Im}(y_e)>0$ for all $e\in E$. Strongly Rayleigh (SR) distributions were defined in \cite{BBL09} where it was shown any $\lambda$-uniform spanning tree distribution is strongly Rayleigh. In this subsection we recall several properties of SR distributions proved in \cite{BBL09,OSS11} which will be useful to us.

\paragraph{Closure Operations.}
Strongly Rayleigh distributions are closed under the following operations.

\begin{itemize}
	\item {\bf Projection.} For any $\mu\in \cB_E$, and any $F\subseteq E$, the projection of $\mu$ onto $F$ is the measure $\mu_F$ where for any $A\subseteq F$,
	$$ \mu|_F(A)=\sum_{S: S\cap F=A} \mu(S).$$
	\item {\bf Conditioning.} For any $e\in E$, $\{\mu | e \text{ out}\}$ and $\{\mu | e \text{ in}\}$.
\end{itemize}

\paragraph{Negative Dependence Properties.}
An {\em increasing function} $f:2^E\rightarrow \R$, is a function where for any $A\subseteq B\subseteq E$, we have $f(A)\leq f(B)$.
For example, if $E$ is the set of edges of a graph $G$, then the existence of a Hamiltonian cycle is an increasing function, and the $3$-colorability of $G$ is a decreasing function.

\begin{definition}[Negative Association]
\label{def:negativeassociation}
A measure $\mu \in {\cal B}_E$ is {\em negatively associated} or NA if
for any increasing functions $f,g: 2^E\to \R$, that depend on {\em disjoint} sets of edges,
$$ \EE{\mu}{f}\cdot \EE{\mu}{g} \geq  \EE{\mu}{f\cdot g} $$
\end{definition}

\noindent
It is shown in \cite{BBL09,FM92}  that strongly Rayleigh measures
are negatively associated. 

\vspace{0.1in}
\noindent
Let $\mu$ be a strongly Rayleigh measure on edges of $G$. For a set $A$, let $$X_A=|A\cap S|$$ be the random variable indicating the number of edges in $A$ chosen in a random sample $S$. The following facts immediately follow from the negative association property and the fact that any tree has exactly $n-1$ edges, see \cite{OSS11} for more details.

\begin{fact}
\label{fact:updown}
If $\mu$ is a $\lambda$-uniform spanning tree distribution on $G=(V,E)$, then
for any $S\subset E$, $p\in \mathbb{R}$

\begin{enumerate}
\item If $e\notin S$, then $\EE{\mu}{X_{e} \big| X_S \geq p} \leq \EE{\mu}{X_{e}}$ and 
$\EE{\mu}{X_{e} \big| X_S \leq p} \geq \EE{\mu}{X_{e}} $
\item If $e\in S$, then $\EE{\mu}{X_e | X_S \geq p} \geq \EE{\mu}{X_e}$ and $\EE{\mu}{X_e | X_S\leq p} \leq \EE{\mu}{X_e}$.
\end{enumerate}
\end{fact}
\begin{fact}
For any set of edges $S$ and $e\not\in S$,
\begin{equation}
\EE{\mu}{X_S} \le \EE{\mu}{ X_S | X_e = 0} \le \EE{\mu}{X_S} + x_e
\label{fact:e0}
\end{equation}
and
\begin{equation}
\EE{\mu}{X_S}  - x_e \le\EE{\mu}{ X_S | X_e = 1} \le \EE{\mu}{X_S}.
\label{fact:e1}
\end{equation}
\end{fact}
\begin{lemma} \label{lem:correlation}
Let $S = \{e_1, \ldots, e_k\}$ be a set of  $k$ edges and suppose that $e \not\in S$. 
Then there are $k-1$ edges in $S$ such that for each edge $e_i$ among these $k-1$
$$0.25  \le \P{e_i \in T | e \in T} \le 0.5.$$
\end{lemma}

\begin{proof}
By \cref{fact:e1}, $\P{e_i | e\in T} \leq 0.5$ for all $e_i$.
Suppose that (after renaming) $\P{e_1 | e\in T}$ and $\P{e_2 | e\in T}$ are the smallest among all $1\leq i\leq k$. Observe that  by \cref{fact:e1} $\E{X_{e_1}+X_{e_2} | e\in T}\geq 0.5$. Therefore, the bigger one is at least $0.25$.
\end{proof}

\paragraph{Rank Sequence.}

The {\em rank sequence} of $\mu$ is the sequence
$$\P{X_E=0}, \P{X_E=1}, \ldots,\P{X_E=m}.$$
Let $g_\mu(y)$ be the generating polynomial of $\mu$.
The {\em diagonal specialization} of $\mu$, $\bar{g}(.)$ is a univariate polynomial obtained by treating $g(.)$ as a univariate polynomial (i.e., considering $g(y,y,\ldots,y)$). Observe that $\bar{g}(.)$ is the generating polynomial of the rank sequence of $\mu$. If $\bar{g}(c)=0$ for $c\in \mathbb{C}$, then $g(c,c,\ldots,c)=0$. So,  if $g(.)$ is a real stable polynomial then so is $\bar{g}$. Since a univariate polynomial with real coefficients is stable if and only if all of its roots are real, $\bar{g}(.)$ is a polynomial with real roots.

     Generating polynomials of probability distributions with real roots are very well studied in the literature (see \cite{Pit97} and references therein). If $\bar{g}(.)$ is a real rooted univariate polynomial of degree $m$ with nonnegative coefficients,  then coefficients of $\bar{g}(.)$  correspond to the probability density function of the convolution of a set of $m$ independent Bernoulli random variables (up to a normalization). In other words,
     there are $m$ independent Bernoulli random variables $B_1,\ldots,B_m$ with success probabilities $p_1,\ldots,p_m\in[0,1]$ such that the probability that exactly $k$ variables succeed
     is the coefficient of $y^k$ in $\bar{g}(.)$.

\begin{fact}[{\cite{BBL09,Pit97}}]
\label{fact:SRsumindependentBernoulli}
The rank sequence of a strongly Rayleigh measure is the probability distribution of the number
of successes in $m$ independent trials for some sequence of success probabilities $p_1,\ldots,p_m\in [0,1]$.
\end{fact}

Given this, we can apply the following theorem by Hoeffding \cite{Hoe56}, following the approach of~\cite{OSS11}.
\begin{theorem}[{\cite[Corollary 2.1]{Hoe56}}]\label{thm:hoeffding}
Let $g:\{0,1,\dots,m\}\to \R$ and $0\leq p\leq m$ for some integer $m\geq 0$.  Let $B_1,\dots,B_m$ be $m$ independent Bernoulli random variables with success probabilities $p_1,\dots,p_m$ that minimizes (or maximizes)
$$ \E{g(B_1+\dots+B_m)}$$
over all distributions in ${\cal B}_m(p)$. Then,  $p_1,\dots,p_m\in\{0,x,1\}$ for some $0<x<1$.
\end{theorem}



\begin{lemma}\label{lem:420} Let $S \subseteq E$ with $|S| = 3$. Furthermore, assume that $\P{|S \cap T| \ge 1} = 1$. Then, $\P{|S \cap T| = 1} \ge \frac{1}{2}$ and $\P{|S \cap T| = 2} \ge \frac{3}{8}$.\end{lemma}
\begin{proof}
	By \autoref{fact:SRsumindependentBernoulli}, we can write the rank sequence of $|S \cap T|$ as a sum of 3 independent Bernoullis $B_1,B_2,B_3$, and since $\P{|S \cap T| \ge 1} = 1$ we know that for one Bernoulli $p=1$. Without loss of generality let $p_1 = 1$. Then by \autoref{thm:hoeffding} we know that $\P{|S \cap T| = 1}$ and $\P{|S \cap T| = 2}$ are minimized when $p_2 = p_3 = \frac{1}{4}$ or $p_2 = \frac{1}{2}$ and $p_3 = 0$. Therefore:
	$$\P{|S \cap T| = 1} \ge \min\bigg\{\bigg(\frac{3}{4}\bigg)^2,\frac{1}{2}\bigg\} = \frac{1}{2}$$
	$$\P{|S \cap T| = 2} \ge \min\bigg\{2\bigg(\frac{1}{4}\bigg)\bigg(\frac{3}{4}\bigg),\frac{1}{2}\bigg\} = \frac{3}{8}$$
\end{proof}


The following two lemmas are proved using a similar analysis.
\begin{lemma}\label{lem:420-2} Let $S \subseteq E$ with $|S| = 2$. Let $\frac{1}{2} \le \E{|S \cap T|} \le \frac{3}{2}$. Then $\P{|S \cap T| = 1} \ge \frac{3}{8}$.\end{lemma}

\begin{lemma} \label{lem:cut_even}
For a min-cut $C$, $\P{|T\cap C|\text{ even}}\geq 13/27$.
\end{lemma}

\begin{lemma}\label{lem:427} Let $S_1,S_2 \subseteq E$ with $|S_1 \cap S_2| = \varnothing$. Let $|S_1|=|S_2|=2$, or equivalently $\E{|S_1 \cap T|} = \E{|S_2 \cap T|} = 1$. Then $\P{|S_1 \cap T| = 1 \land |S_2 \cap T| = 1} \ge \frac{3}{16}$.\end{lemma}
\begin{proof}
	Let $S_1 = \{e,f\}$. Then condition on $e \in T$: this occurs with probability $\frac{1}{2}$. By \autoref{fact:updown} we have 
	$$\E{|f \cap T| \mid e \in T} \le \frac{1}{2}$$
	Then condition on $f \not\in T$. Given the above, this happens with probability at least $\frac{1}{2}$. Similarly consider the event $e \not\in T$ and $f \in T$. One of these occurs with probability $\frac{1}{2}$. Therefore, in either event we have:
	$$\frac{1}{2} \le \E{|S_2 \cap T|} \le \frac{3}{2}$$
	And now by \autoref{lem:420-2} both events occur simultaneously with probability at least $\frac{3}{16}$.
\end{proof}

\section{Overview of Analysis}

As already mentioned, our algorithm consists of two steps: sample a 1-tree $T$, and then construct an optimal $O$-join for the odd degree vertices in the 1-tree.

Given a feasible LP solution $x$, the choice $y_e = x_e/2$ for each edge $e\in E$ (which gives $y_e := 1/4$ in the half integral case), yields an $O$-join solution of total cost at most $OPT/2$. However, this is essentially Christofides' algorithm and
 guarantees only a 3/2 approximation. 

The key to improving on this is the observation that constraint
\eqref{eq:tjoinlp} in the $O$-join LP is not binding if the intersection of the cut $\delta(S)$ with the tree is even. 

\begin{definition}[Even cuts]
A cut $\delta(S)$ is \textbf{even} in $T$ (or simply ``even'' when T is understood) if $|T \cap \delta(S)|$ is even.
\end{definition}

Thus, for every edge $e$ with the property that every min-cut that $e$ is on is even, we can reduce $y_e$ to 1/6, since every non-min-cut has at least 6 edges, and therefore this guarantees that constraint \eqref{eq:tjoinlp} remains satisfied everywhere. This is the gist of the approach taken in \cite{OSS11}.








Suppose that there are multiple ``good" edges $e$ with the property that every min-cut they are on is even, say with probability at least $p$ (over the randomness in the selection of $T$). Then for those outcomes $T$ in which $e$ has this property, we could set $y_e := 1/6$  and satisfy the $O$-join constraints. This would save us $\frac{1}{12}c(e)$ on every such edge $e$  (the reduction from $1/4$ to $1/6$) and thereby 
guarantee a reduction in the cost of the $O$-join solution
of $\sum_{e\text{ "good"}} \frac{p}{12} c(e)$. 

Unfortunately, in general, it is not possible to 
argue that every min-cut an edge is on is even simultaneously in $T$ with constant probability. So, we will use a careful charging scheme.


\begin{definition}[Last Cuts]
For an edge $e$, the last cuts of $e$ are the only ({\em two}) min-cuts containing $e$ and edges going higher in the graph right before contracting $S_e$.
\end{definition}

Observe that the last cuts of a top edge are critical cuts, but the last cuts of bottom edges are not critical. 

\begin{definition}[Even at Last]
For an edge $e$ we say $e$ is even at last if the two last cuts of $e$ are even.

Equivalently, if $e$ is a bottom edge, we say $e$ is even at last if all the min cuts containing $e$ on the cycle defined by the graph consisting of  $S_e$ with $V \smallsetminus S_e$ contracted are even.
Otherwise, if $e=\{u,v\}$ is a top edge, then it is even at last if the critical cuts $S_u,S_v$ are even simultaneously.
\end{definition}

\begin{fact}
\label{fact:criticalsat}
If a bottom edge $e$ is even at last, then all (bottom) edges $f$ where $S_f=S_e$ are even at last.
\end{fact}
\begin{proof}
Since $\delta(S_e)$ is a cycle cut, the edges inside $S_e$ form a path,
and thus, exactly one edge between each pair  of (possibly contracted) vertices inside $S_e$ is selected as part of the tree on $S_e$ chosen in step 6 of the algorithm. 
If, $e$ is even at last, we must have  exactly one of each pair of cycle partners on $\delta(S)$
is in $T$; therefore, every pair of adjacent nodes in the cycle have one edge connecting them in the tree. So,  all cuts on the cycle have exactly two edges in $T$. This implies every bottom edge of this cycle is even at last. 
\end{proof}

\begin{remark}
By \cref{fact:min-cuts2}, the companion of every bottom edge $e$ has exactly the same pair of last cuts as $e$. 
\end{remark}

\begin{definition}[Good edges]\label{defn:good}
An edge $e=(u,v)$ is \textbf{good} if it is even at last with probability at least $p$ for some constant $p > 0$. 
\end{definition}



Instead of proving that all min-cuts that a single edge is on are even, we will instead prove that every minimum cut contains at least one good edge. Each good edge $e$ will then be responsible for its last two cuts. This will allow edges to be reduced when they are even at last, as all cuts lower in the hierarchy are handled by other edges.
\begin{theorem}\label{thm:mainprobabilistic}
There is a universal constant $p\geq 1/27$ such that every every min-cut has at least one good edge. 
\end{theorem}
The proof of the above theorem together with some strengthened statements will be in \cref{sec:probFacts}. The proof mainly exploits properties of strongly Rayleigh distributions.


As we hinted at, we will reduce the value of $y_e$ to 1/6 whenever an edge $e$ is even at last. However, since $e$ may also be on many other lower min-cuts, if we reduce $y_e$, the solution may not be feasible (\eqref{eq:tjoinlp} may be violated) as the lower min-cuts may be odd. To handle any lower min-cut $C$ that $e$ is on,  we show that, conditioned on $e$ being even at last, the probability that $C$ is also even is at least $q$ for some $q\geq \Omega(1)$. Therefore, we only need to worry about the lower cuts with probability $1-q$ each.   In the bad event that a lower cut $C$ is odd, we will need to fix the solution to guarantee that \eqref{eq:tjoinlp} still holds: our approach is to split the deficit introduced in the $O$-Join constraint for $C$ among the good edges that do not go higher (see \cref{def:gohigher}). We then simply show that in expectation each edge gains. 
This part of the proof heavily exploits the properties of cactus representation of the min-cuts that we discussed above.


We note that \cref{thm:mainprobabilistic} on its own is not enough to run our charging argument; so, we need a slightly stronger version. In particular, in some cuts we may need to have two or three good edges. 

